\documentclass{llncs} 

\usepackage{amssymb}
\usepackage{textcomp}
\usepackage{amsmath}
\usepackage{algorithm}
\usepackage{algpseudocode}
\usepackage{comment}
\usepackage{graphicx}
\usepackage{color}
\newcommand{\Endproof}{\hfill$\Box$\\}

\begin{document}

\title{Quantum Algorithm for Dynamic Programming Approach for DAGs. Applications for Zhegalkin Polynomial Evaluation and Some Problems on DAGs}
\author{Kamil~Khadiev$^{1,2}$, Liliya Safina$^{2}$} 

\institute{Smart Quantum Technologies Ltd., Kazan, Russia\and  Kazan Federal University, Kazan, Russia \\ \email{kamilhadi@gmail.com, liliasafina94@gmail.com} 
}

\maketitle

\begin{abstract}
In this paper, we present a quantum algorithm for dynamic programming approach for problems on directed acyclic graphs (DAGs). The running time of the algorithm is $O(\sqrt{\hat{n}m}\log \hat{n})$, and the running time of the best known deterministic algorithm is $O(n+m)$, where $n$ is the number of vertices, $\hat{n}$ is the number of vertices with at least one outgoing edge; $m$ is the number of edges. We show that we can solve problems that use OR, AND, NAND, MAX and MIN functions as the main transition steps. The approach is useful for a couple of problems. One of them is computing a Boolean formula that is represented by Zhegalkin polynomial, a Boolean circuit with shared input and non-constant depth evaluating. Another two are the single source longest paths search for weighted DAGs and the diameter search problem for unweighted DAGs.

\end{abstract}

\section{Introduction}
\emph{Quantum computing} \cite{nc2010,a2017} is one of the hot topics in computer science of last decades.
There are many problems where quantum algorithms outperform the best known classical algorithms \cite{dw2001,quantumzoo}. 
superior of quantum  over classical was shown for different computing models like query model, streaming processing models, communication models and others \cite{an2009,av2009,agky14,agky16,aakk2018,aakv2018,kkm2018,kk2017,ikpy2018,l2009}.

In this paper, we present the quantum algorithm for the class of problems on directed acyclic graphs (DAGs) that uses a dynamic programming approach.
The dynamic programming approach is one of the most useful ways to solve problems in computer science \cite{cormen2001}. The main idea of the method is to solve a problem using pre-computed solutions of the same problem, but with smaller parameters. Examples of such problems for DAGs that are considered in this paper are the single source longest path search problem for weighted DAGs and the diameter search problem for unweighted DAGs.
 
 Another example is a Boolean circuit with non-constant depth and shared input evaluation. A Boolean circuit can be represented as a DAG with conjunction (AND) or disjunction (OR) in vertices, and inversion (NOT) on edges. We present it as an algorithm for computing a Zhegalkin polynomial \cite{z1927,gg85,y2003}. The Zhegalkin polynomial is a way of a Boolean formula representation using ``exclusive or'' (XOR, $\oplus$), conjunction (AND) and the constants $0$ and $1$. 

The best known deterministic algorithm for dynamic programming on DAGs uses depth-first search algorithm (DFS) as a subroutine \cite{cormen2001}. Thus, this algorithm has at least depth-first search algorithm's time  complexity, that is $O(n+m)$, where $m$ is the number of edges and $n$ is the number of vertices. The query complexity of the algorithm is at least $O(m)$.

We suggest a quantum algorithm with the running time $O(\sqrt{\hat{n}m}\log \hat{n})$, where $\hat{n}$ is the number of vertices with non-zero outgoing degree. In a case of $\hat{n}(\log \hat{n})^2<m$, it shows speed-up comparing with deterministic algorithm. The quantum algorithm can solve problems that use a dynamic programming algorithm with OR, AND, NAND, MAX or MIN functions as transition steps. We use Grover's search  \cite{g96,bbht98} and D{\"u}rr and H{\o}yer maximum search \cite{dh96} algorithms to speed up our search. A similar approach has been applied by D{\"u}rr et al. \cite{dhhm2004,dhhm2006}; Ambainis and {\v{S}}palek \cite{as2006}; D{\"o}rn \cite{d2009,d2008phd} to several graph problems.

We apply this approach to four problems that discussed above. 
The first of them involves computing Boolean circuits. Such circuits can be represented as AND-OR-NOT DAGs. Sinks of such graph are associated with Boolean variables, and other vertices are associated with conjunction (AND) or a disjunction (OR); edges can be associated with inversion (NOT) function. Quantum algorithms for computing  AND-OR-NOT trees were considered by Ambainis et al. \cite{avrz2010,a2007,a2010}. Authors present an algorithm with running time $O(\sqrt{N})$, where $N$ is the number of tree's vertices. There are other algorithms that allow as construct AND-OR-NOT DAGs of constant depth, but not a tree \cite{bkt2018,ckk2012}.

Our algorithm works with $O(\sqrt{\hat{n}m}\log \hat{n})$ running time for DAGs that can have non-constant depth.

It is known that any Boolean function can be represented as a Zhegalkin polynomial \cite{z1927,gg85,y2003}. The computing of a Zhegalkin polynomial is the second problem.  Such formula can be represented as an AND-OR-NOT DAG. Suppose an original Zhegalkin polynomial has $t_c$ conjunctions and $t_x$ exclusive-or operations. Then the corresponding AND-OR-NOT DAG have $\hat{n}=O(t_x)$ non-sink vertices and $m=O(t_c+t_x)$ edges. 

If we consider AND-OR-NOT trees representation of the formula, then it has an exponential number of vertices $N\geq 2^{O(t_x)}$. The quantum algorithm for trees \cite{avrz2010,a2007,a2010} works in  $O(\sqrt{N})=2^{O(t_x)}$ running time. Additionally, the DAG that corresponding to a Zhegalkin polynomial has non-constant depth. Therefore, we cannot use algorithms from \cite{bkt2018,ckk2012} that work  for circuits with shared input.

The third problem is the single source longest path search problem for a weighted DAG. The best deterministic algorithm for this problem works in $O(n+m)$ running time \cite{cormen2001}. In the case of a general graph (not DAG), it is a NP-complete problem.  Our algorithm for DAGs works in $O(\sqrt{\hat{n}m}\log \hat{n})$ running time. The fourth problem is the diameter search problem for an unweighted DAG. The best deterministic algorithms for this problem works in $O(\hat{n}(n+m))$ expected running time \cite{cormen2001}. Our algorithm for DAGs works in expected $O(\hat{n}(n+\sqrt{nm})\log n)$ running time.

  The paper is organized as follows. We present definitions in Section \ref{sec:prlmrs}. Section  \ref{sec:algorithm} contains a general description of the algorithm. The application to an AND-OR-NOT DAG evaluation and Zhegalkin polynomial evaluation is in Section \ref{sec:andor-dag}. Section  \ref{sec:diameter} contains a solution for the single source longest path search problem for a weighted DAG and the diameter search problem for an unweighted DAG.

\section{Preliminaries}\label{sec:prlmrs}
Let us present definitions and notations that we use.

A graph $G$ is a pair $G=(V,E)$ where $V=\{v_1,\dots,v_n\}$ is a set of vertices, $E=\{e_1,\dots,e_m\}$ is a set of edges, an edge $e\in E$ is a pair of vertices $e=(v,u)$, for $u,v\in V$.
A graph $G$ is directed if all edges $e=(v,u)$ are ordered pairs, and if $(v,u)\in E$, then $(u,v)\not\in E$. In that case, an edge $e$ leads from vertex $v$ to vertex $u$. A graph $G$ is acyclic if there is no path that starts and finishes in the same vertex.
We consider only directed acyclic graphs (DAGs) in the paper.

Let $D_i=(v:\exists e=(v_i,v)\in E)$ be a list of  $v_i$ vertex's out-neighbors. Let $d_i=|D_i|$ be the out-degree of the vertex $v_i$.
 Let $L$ be a set of indices of sinks. Formally, $L = \{i : d_i =0, 1 \leq i \leq n\}$. Let $\hat{n}=n-|L|$.
 
Let $D'_i=(v:\exists e=(v,v_i)\in E)$ be a list of vertex whose  out-neighbor is $v$. 
 
 We consider only DAGs with two additional properties:
 \begin{itemize}
 \item topological sorted:  if there is an edge $e=(v_i,v_j)\in E$, then $i<j$;
 \item last $|L|$ vertices belong to $L$, formally $d_i=0$, for $i>\hat{n}$.
 \end{itemize}

Our algorithms use some quantum algorithms as a subroutine, and the rest is classical. As quantum algorithms, we use query model algorithms. These algorithms can do a query to black box that has access to the graph structure and stored data. As a running time of an algorithm we mean a number of queries to black box. We use an {\em adjacency list model} as a model for a graph representation. The input is specified by $n$ arrays $D_i$, for $i\in\{1\dots n\}$. 
We suggest \cite{nc2010} as a good book on quantum computing.


\section{Quantum Dynamic Programming Algorithm for DAGs}\label{sec:algorithm}

Let us describe an algorithm in a general case. 

Let us consider some problem $P$ on a directed acyclic graph $G=(V,E)$.
Suppose that we have a dynamic programming algorithm for $P$ or we can say that there is a solution of the problem $P$ that is equivalent to computing a function $f$ for each vertex.
As a function $f$ we consider only functions from a set ${\cal F}$ with following properties:
\begin{itemize}
\item $f:V\to \Sigma$.
\item The result set $\Sigma$ can be the set of real numbers $\mathbb{R}$, or integers $\{0,\dots,{\cal Z}\}$, for some integer ${\cal Z}>0$.
\item if $d_i>0$ then $f(v_i)=h_i(f(u_1),\dots,f(u_{d_i}))$, where functions $h_i$ are such that $h_i:\Sigma^{[1,n]}\to \Sigma$; $\Sigma^{[1,n]}=\{(r_1,\dots,r_k): r_j\in\Sigma, 1\leq j \leq k, 1\leq k \leq n\}$ is a set of vectors of at most $n$ elements from $\Sigma$; $(u_1,\dots,u_{d_i})=D_i$.
\item if $d_i=0$ then  $f(v_i)$ is classically computable in constant time. 
\end{itemize}

Suppose there is a quantum algorithm $Q_i$ that computes function $h_i$ with running time $T(k)$, where $k$ is a length of the argument for the function $h_i$. Then we can suggest the following algorithm:

Let $t=(t[1],\dots,t[\hat{n}])$ be an array which stores results of the function $f$. Let $t_f(j)$ be a function such that $t_f(j)=t[j]$, if $j\leq \hat{n}$; $t_f(j)=f(j)$, if $j> \hat{n}$. Note that $j> \hat{n}$ means $v_j\in L$.
\begin{algorithm}
\caption{Quantum Algorithm for Dynamic programming approach on DAGs.}\label{alg:general-dp}
\begin{algorithmic}
\For{$i=\hat{n}\dots 1$}
  \State $t[i] \gets Q_i(t_f(j_1),\dots,t_f(j_{d_i}))$, where $(v_{j_1},\dots,v_{j_{d_i}})=D_i$
\EndFor
\State \Return $t[1]$
\end{algorithmic}
\end{algorithm}

Let us discuss the running time of Algorithm \ref{alg:general-dp}. The proof is simple, but we present it for completeness. 
\begin{lemma}\label{lm:dp-time}
Suppose that the quantum algorithms $Q_i$ works in $T_i(k)$ running time, where $k$ is a length
of an argument, $i\in\{1,\dots,n\}$. Then Algorithm \ref{alg:general-dp} works in $T^1 = \sum\limits_{i\in\{1,\dots,n\}\backslash L}
T_i(d_i)$.
\end{lemma}
\begin{proof}
Note, that when we compute $t[i]$, we already have computed $t_f(j_1),\dots,t_f(j_{d_i})$ or can compute them in constant time because for all $e=(v_i,v_j)\in E$ we have $i<j$. 

A complexity of a processing a vertex $v_i$ is $T_i(d_i)$, where $i \in \{1, \dots , n\}\backslash L$.

The algorithm process vertices one by one. Therefore 
\[T^1 = \sum\limits_{i\in\{1,\dots,n\}\backslash L}
T_i(d_i).\]
\Endproof
\end{proof}

Note, that quantum algorithms have a probabilistic behavior. Let us compute an error probability for
Algorithm \ref{alg:general-dp}.

\begin{lemma}\label{lm:dp-err}
Suppose the quantum algorithm $Q_i$ for the function $h_i$ has an error probability $\varepsilon(n)$, where $i\in\{1,\dots,n\}\backslash L$. Then the error probability of Algorithm \ref{alg:general-dp} is at most $1 - (1 - \varepsilon(n))^{\hat{n}}$.
\end{lemma}
\begin{proof}
Let us compute a success probability for Algorithm \ref{alg:general-dp}. Suppose that all vertices are
computed with no error. The probability of this event is $(1 - \varepsilon(n))^{\hat{n}}$, because error of each invocation is independent event.

Therefore, the error probability for Algorithm \ref{alg:general-dp} is at most $1-(1-\varepsilon(n))^{\hat{n}}$, for $\hat{n}=n-|L|$. 
\Endproof
\end{proof}

For some functions and algorithms, we do not have a requirement that all arguments of $h$ should be computed with no error. In that case, we will get better error probability. This situation is discussed in Lemma \ref{lm:dp-max-compl}.

\subsection{Functions for Vertices Processing}
We can choose the following functions as a function $h$
\begin{itemize}
\item Conjunction ($AND$ function). For computing this function, we can use Grover's search
algorithm \cite{g96,bbht98} for searching $0$ among arguments. If the element that we found is $0$, then the result is $0$. If the element is $1$, then there is no $0$s, and the result is $1$.
\item Disjunction ($OR$ function). We can use the same approach, but here we search $1$s.
\item Sheffer stroke (Not $AND$ or $NAND$ function). We can use the same approach as for AND function, but here we search $1$s. If we
found $0$ then the result is $1$; and $0$, otherwise.
\item Maximum function ($MAX$). We can use the D{\"u}rr and H{\o}yer maximum search algorithm \cite{dh96}.
\item Minimum function ($MIN$). We can use the same algorithm as for maximum.
\item Other functions that have quantum algorithms.
\end{itemize}

As we discussed before, $AND$, $OR$ and $NAND$ functions can be computed using Grover search algorithm. Therefore algorithm for these functions on vertex $v_i$ has an error $\varepsilon_i\leq 0.5$ and running time is $T(d_i)=O(\sqrt{d_i})$, for  $i\in \{1,\dots,n\}\backslash L$. These results follow from \cite{bhmt2002,gr2005,bbht98,g96} We have the similar situation for computing maximum and minimum functions \cite{dh96}.

If we use these algorithms in Algorithm \ref{alg:general-dp} then we obtain the error probability  $1-(0.5)^{\hat{n}}$ due to Lemma \ref{lm:dp-err}.

At the same time, the error is one side. That is why we can apply the boosting technique to reduce the error probability. The modification is presented in Algorithm \ref{alg:bost-qalgo}.

Let us consider $MAX$ and $MIN$ functions. Suppose, we have a quantum algorithm $Q$ and a deterministic algorithm $A$ for a function $h\in\{MAX,MIN\}$. Let $k$ be number of algorithm's invoking according to the boosting technique for reducing the error probability. The number $k$ is integer and $k\geq 1$. Let $(x_1,\dots,x_d)$ be arguments (input data) of size $d$.

Let us denote it as $\hat{Q}^k(x_1,\dots,x_d)$. Suppose, we have a temporary array $b=(b[1],\dots,b[k])$.
\begin{algorithm}
\caption{The application of the boosting technique to Quantum Algorithm.}\label{alg:bost-qalgo}
\begin{algorithmic}
\For{$z=1\dots k$}
  \State $b[z] \gets Q(x_1,\dots,x_d)$
\EndFor
\State $result \gets  A_i\left(b[1],\dots, b[k]\right)$
\State \Return $result$
\end{algorithmic}
\end{algorithm}

If we analyze the algorithm, then we can see that it has the following property:
\begin{lemma} \label{lm:boosting}
 Let $(x_1,\dots,x_d)$ be an argument (input data) of size $d$, for a function $h(x_1,\dots,x_d)\in\{MAX,MIN\}$. Let $k$ be a number of algorithm's invoking. The number $k$ is integer and $k\geq 1$.   Then the expected running time of the boosted version $\hat{Q}^k(x_1,\dots,x_d)$ of the quantum algorithm $Q(x_1,\dots,x_d)$ is $O\left(k\cdot \sqrt{d}\right)$ and the error probability is $O(1/2^{k})$.
\end{lemma}
\begin{proof}
Due to \cite{dh96}, the expected running time of the algorithm $Q$ is $O\left(\sqrt{d}\right)$ and the error probability is at most $0.5$. We repeat the algorithm $k$ times and then apply $A$ (function $MAX$ or $MIN$) that has running time $O(k)$. Therefore, the expected running time is $O\left(k\sqrt{d}\right)$. The algorithm is successful if at least one invocation is successful because we apply $A$. Therefore, the error probability is $O\left(1/2^{k}\right)$.
\Endproof
\end{proof}

At the same time, if an algorithm $Q$ measures only in the end, for example Grover Algorithm; then we can apply the amplitude amplification algorithm \cite{bhmt2002} that boosts the quantum algorithm in running time $O(\sqrt{k})$. The amplitude amplification algorithm is the generalization of the Grover's search algorithm. So, in the case of $AND, OR$ and $NAND$ functions, we have the following property: 

\begin{lemma} \label{lm:boosting-q}
 Let $(x_1,\dots,x_d)$ be an argument (input data) of size $d$ for some function $h(x_1,\dots,x_d)\in\{AND,OR,NAND\}$. Suppose, we have a quantum algorithm $Q$ for $h$. Let $k$ be planned number of algorithm's invoking in a classical case. The number $k$ is integer and $k\geq 1$. Then the running time of the boosted version $\hat{Q}^k(x_1,\dots,x_d)$ of the quantum algorithm $Q(x_1,\dots,x_d)$ is $O\left(\sqrt{k\cdot d}\right)$ and the error probability is $O\left(1/2^{k}\right)$.
\end{lemma}
\begin{proof}
Due to \cite{bhmt2002,gr2005,bbht98,g96}, the running time of the algorithm $Q$ is at most $O\left(\sqrt{d}\right)$ and the error probability is at most $0.5$. We apply amplitude amplification algorithm \cite{bhmt2002} to the algorithm and gets  the claimed running time and the error probability
\Endproof
\end{proof}

Let us apply the previous two lemmas to Algorithm \ref{alg:general-dp} and functions from the set  $\{AND,OR,NAND,MAX,MIN\}$.

\begin{lemma}\label{lm:dp-general-complx}

Suppose that a problem $P$ on a DAG $G=(V,E)$ has a dynamic programming algorithm such that functions  $h_i\in\{AND,OR,NAND,MAX,MIN\}$, for $i\in\{1,\dots,\hat{n}\}$. Then there is a quantum dynamic programming algorithm $A$ for the problem $P$ that has running time $O(\sqrt{\hat{n}m}\log \hat{n})=O(\sqrt{nm}\log n)$ and error probability $O(1/\hat{n})$. Here $m=|E|, n=|V|, \hat{n}=n-|L|$, $L$ is the set of sinks.
\end{lemma}
\begin{proof}

Let us choose $k=2\log_2 \hat{n}$ in Lemmas \ref{lm:boosting},\ref{lm:boosting-q}. Then the error probabilities for the algorithms $Q^{2\log_2 \hat{n}}$ are $O\left(0.5^{2\log_2 \hat{n}}\right)=O\left(1/n^2\right)$. The  running time is $O(\sqrt{d_i\log \hat{n}})=O(\sqrt{d_i}\log \hat{n})$ for $h_i\in\{AND,OR,NAND\}$ and  $O(\sqrt{d_i}\log \hat{n})$ for $h_i\in\{MAX,MIN\}$.

Due to Lemma \ref{lm:dp-err}, the probability of error is at most $\varepsilon(\hat{n})= 1-\left(1-\frac{1}{\hat{n}^2}\right)^{\hat{n}}$. Note that 
\[\lim\limits_{\hat{n}\to \infty}\frac{\varepsilon(\hat{n})}{1/\hat{n}}= \lim\limits_{\hat{n}\to \infty} \frac{1-\left(1-\frac{1}{\hat{n}^2}\right)^{\hat{n}}}{1/\hat{n}}=1;\]
 Hence, $\varepsilon(\hat{n})=O(1/\hat{n})$.

Due to Lemma  \ref{lm:dp-time}, the running time is
\[T^1=\sum\limits_{i\in\{1,\dots,n\}\backslash L}
T_i(d_i) \leq \sum\limits_{i\in\{1,\dots,n\}\backslash L}
O\left(\sqrt{d_i}\log \hat{n}\right)  =O\left((\log_2 \hat{n})\cdot\sum\limits_{i\in\{1,\dots,n\}\backslash L}
\sqrt{d_i}\right).\]
Due to the Cauchy-Bunyakovsky-Schwarz inequality, we have 
\[\sum\limits_{i\in\{1,\dots,n\}\backslash L} \sqrt{d_i}\leq \sqrt{\hat{n}\sum\limits_{i\in\{1,\dots,n\}\backslash L}d_i}\]
Note that $d_i=0$, for $i\in L$. Therefore, $\sum\limits_{i\in\{1,\dots,n\}\backslash L}d_i=\sum\limits_{i\in\{1,\dots,n\}}d_i=m$, because $m=|E|$ is the total number of edges. Hence,
\[\sqrt{\hat{n}\sum\limits_{i\in\{1,\dots,n\}\backslash L}d_i}=\sqrt{\hat{n}\sum\limits_{i\in\{1,\dots,n\}}d_i}=\sqrt{\hat{n}m}.\] Therefore, $T_1\leq O(\sqrt{\hat{n}m}\log \hat{n})= O(\sqrt{nm}\log n)$.   
\Endproof
\end{proof}

If $h_i\in\{AND,OR,NAND\}$, then we can do a better estimation of the running time. 

\begin{lemma}\label{lm:dp-bool-complx}
Suppose that a problem $P$ on a DAG $G=(V,E)$ has a dynamic programming algorithm such that functions  $h_i\in\{AND,OR,NAND\}$, for $i\in\{1,\dots,\hat{n}\}$. Then there is a quantum dynamic programming algorithm $A$ for the problem $P$ that has running time $O(\sqrt{\hat{n}m\log \hat{n}})=O(\sqrt{nm\log n})$ and error probability $O(1/\hat{n})$. Here $m=|E|, n=|V|, \hat{n}=n-|L|$, $L$ is the set of sinks.
\end{lemma}
\begin{proof}
The proof is similar to the proof of Lemma \ref{lm:dp-general-complx}, but we use $O(\sqrt{d_i\log \hat{n}})$ as time complexity of processing a vertex.
\Endproof
\end{proof}

If $h_i\in\{MAX,MIN\}$, then we can do a better estimation of the running time
\begin{lemma}\label{lm:dp-max-compl}
Suppose that a problem $P$ on a DAG $G=(V,E)$ has a dynamic programming algorithm such that functions  $h_i\in\{MAX, MIN\}$, for $i\in\{1,\dots,\hat{n}\}$ and the solution is $f(v_a)$ for some $v_a\in V$. Then there is a quantum dynamic programming algorithm $A$ for the problem $P$ that has expected running time $O(\sqrt{\hat{n}m}\log q)=O(\sqrt{nm}\log n)$ and error probability $O(1/q)$, where $q$ is the length of path to the farthest vertex from the vertex $v_a$. Here $m=|E|, n=|V|, \hat{n}=n-|L|$, $L$ is the set of sinks.
\end{lemma}
\begin{proof}
Let $Q$ be the D{\"u}rr-H{\o}yer quantum algorithm for $MAX$ or $MIN$ function. Let $\hat{Q}^q$ from Algorithm \ref{alg:bost-qalgo} be the boosted version of $Q$. Let us analyze the algorithm.

Let us consider a vertex $v_i$ for $i\in\{1,\dots,n\}\backslash L$. When we process $v_i$, we should compute $MAX$ or $MIN$ among $t_f(j_1),\dots,t_f(j_{d_i})$. Without limit of generalization we can say that we compute MAX function. Let $r$ be an index of maximal element. It is required to have no error for computing $t[j_r]$. At the same time, if we have an error on processing $v_{j_w}$, $w\in\{1,\dots,d_i\}\backslash\{r\}$; then we get value $t[j_w]<f(v_{j_w})$. In that case, we still have $t[j_r]>t[j_w]$. Therefore, an error can be on processing of any vertex $v_{j_w}$.

Let us focus on the vertex $v_a$. For computing $f(v_a)$ with no error, we should compute $f(v_{a_1})$ with no error. Here $v_{a_1}\in D_a$ such that maximum is reached on $v_{a_1}$. For computing $f(v_{a_1})$ with no error, we should compute $f(v_{a_2})$ with no error. Here $v_{a_2}\in D_{a_1}$ such that maximum is reached on $v_{a_2}$ and so on. Hence, for solving problem with no error, we should process only at most  $q$ vertices with no error.  

Therefore, the probability of error for the algorithm is
\[1-\left(1-\left(\frac{1}{2}\right)^{2\log q}\right)^{q}=O\left(\frac{1}{q}\right)\textrm{ because } \lim_{q\to\infty}\frac{1-\left(1-\frac{1}{q^2}\right)^{q}}{1/q}=1.\]
\Endproof
\end{proof}
\section{Quantum Algorithms for Evolution of Boolean Circuits with Shared Inputs and Zhegalkin Polynomial}\label{sec:andor-dag}
Let us apply ideas of quantum dynamic programming algorithms on DAGs to AND-OR-NOT DAGs. 

It is known that any Boolean function can be represented as a Boolean circuit with AND, OR and NOT gates \cite{y2003,ab2009}. Any such circuit can be represented as a DAG with the following properties:
\begin{itemize}
\item sinks are labeled with variables. We call these vertices ``variable-vertices''.
\item There is no vertices $v_i$ such that $d_i=1$.
\item If a vertex $v_i$ such that $d_i\geq 2$; then the vertex labeled with Conjunction or Disjunction.  We call these vertices ``function-vertices''.
\item Any edge is labeled with $0$ or $1$.
\item There is one particular root vertex $v_s$.
\end{itemize}

The graph represents a Boolean function that can be evaluated in the following way. 
We associate a value $r_i\in\{0,1\}$ with a vertex $v_i$, for $i\in\{1,\dots,n\}$. If $v_i$ is a variable-vertex, then $r_i$ is a value of a corresponding variable. If $v_i$ is a function-vertex labeled by a function $h_i\in\{AND, OR\}$, then $r_i=h_i\left(r_{j_1}^{\sigma(i,j_1)},\dots,r_{j_{w}}^{\sigma(i,j_w)}\right)$, where $w=d_i$, $(v_{j_1},\dots,v_{j_w})=D_i$, $\sigma(i,j)$ is a label of an edge $e=(i,j)$. Here,  we say that $x^1=x$ and $x^0=\neg x$ for any Boolean variable $x$. The result of the evolution is $r_s$.

An AND-OR-NOT DAG can be evaluated using the following algorithm that is a modification of Algorithm \ref{alg:general-dp}:

Let $r=(r_1,\dots,r_n)$ be an array which stores results of functions $h_i$. Let a variable-vertex $v_i$ be labeled by $x(v_i)$, for all $i\in L$. Let $Q_i$ be a quantum algorithm for $h_i$; and $\hat{Q}_i^{2\log_2 \hat{n}}$ be a boosted version of $Q_i$ using amplitude amplification (Lemma \ref{lm:boosting-q}).  Let $t_f(j)$ be a function such that $t_f(j)=r_j$, if $j\leq \hat{n}$; $t_f(j)=x(v_j)$, if $j> \hat{n}$. 
\begin{algorithm}
\caption{Quantum Algorithm for AND-OR-NOT DAGs evaluation.}\label{alg:and-or}
\begin{algorithmic}
\For{$i=\hat{n}\dots s$}
  \State $t[i] \gets \hat{Q}_i^{2\log_2 \hat{n}}(t_f(j_1)^{\sigma(i,j_1)},\dots,t_f(j_w)^{\sigma(i,j_w)})$, where $w=d_i$, $(v_{j_1},\dots,v_{j_{w}})=D_i$.
\EndFor
\State \Return $t[s]$
\end{algorithmic}
\end{algorithm}

Algorithm \ref{alg:and-or} has the following property:
\begin{theorem}\label{th:and-or-dag}
Algorithm \ref{alg:and-or} evaluates a AND-OR-NOT DAG $G=(V,E)$ with running time $O(\sqrt{\hat{n}m\log \hat{n}})=O(\sqrt{nm\log n})$ and error probability $O(1/\hat{n})$. Here $m=|E|, n=|V|, \hat{n}=n-|L|$, $L$ is the set of sinks.
\end{theorem}
\begin{proof}
Algorithm \ref{alg:and-or} evaluates the AND-OR-NOT DAG $G$ by the definition of AND-OR-NOT DAG for the Boolean function $F$.
Algorithm \ref{alg:and-or} is almost the same as Algorithm \ref{alg:general-dp}. The difference is labels of edges. At the same time, the Oracle gets information on edge in constant time. Therefore, the running time and the error probability of $\hat{Q}_i^{2\log_2 \hat{n}}$ does not change. Hence, using the proof similar to the proof of Lemma \ref{lm:boosting-q} we obtain the claimed running time and error probability.
\Endproof
\end{proof}

Another way of a Boolean function representation is a NAND DAG or Boolean circuit with NAND gates. \cite{y2003,ab2009}.
We can present NAND formula as a DAG with similar properties as AND-OR-NOT DAG, but function-vertices has only NAND labels. 
At the same time, if we want to use more operations, then we can consider NAND-NOT DAGs and NAND-AND-OR-NOT DAGs:
\begin{theorem}\label{th:nand-dag}
Algorithm \ref{alg:and-or} evaluates a NAND-AND-OR-NOT DAG and a NAND-NOT DAG. If we consider a DAG $G=(V,E)$, then these algorithms work with running time $O(\sqrt{\hat{n}m\log \hat{n}})=O(\sqrt{nm\log n})$ and error probability $O(1/\hat{n})$. Here $m=|E|, n=|V|, \hat{n}=n-|L|$, $L$ is the set of sinks.
\end{theorem}
\begin{proof}
The proof is similar to proofs of Lemma \ref{lm:dp-general-complx} and Theorem \ref{th:and-or-dag}.
\end{proof}

Theorems \ref{th:and-or-dag},\ref{th:nand-dag} present us quantum algorithms for Boolean circuits with shared input and non-constant depth. At the same time, existing algorithms \cite{avrz2010,a2007,a2010,bkt2018,ckk2012} are not applicable in a case of shared input and non-constant depth.

The third way of representation of Boolean function is Zhegalkin polynomial that is representation using $AND, XOR$ functions and the $0,1$ constants \cite{z1927,gg85,y2003}: for some integers $k,t_1,\dots,t_k$,
\[F(x)=ZP(x)=a\oplus\bigoplus_{i=1}^{k}C_i\mbox{, where }a\in\{0,1\}, C_i=\bigwedge_{z=1}^{t_i} x_{j_z}\]

At the same time, it can be represented as an AND-OR-NOT DAG with a logarithmic depth and shared input or an AND-OR-NOT tree with an exponential number of vertices and separated input. That is why the existing algorithms from \cite{avrz2010,a2007,a2010,bkt2018,ckk2012} cannot be used or work in exponential running time. 

\begin{theorem}\label{th:xor-dag}
Algorithm \ref{alg:and-or} evaluates the XOR-AND DAG $G=(V,E)$ with running time $O(\sqrt{\hat{n}m\log \hat{n}})=O(\sqrt{nm\log n})$ and error probability $O(1/\hat{n})$. Here $m=|E|, n=|V|, \hat{n}=n-|L|$, $L$ is the set of sinks.
\end{theorem}
\begin{proof}
 $XOR$ operation is replaced by two $AND$, one $OR$ vertex and $6$ edges because for any Boolean $a$ and $b$ we have $a \oplus b = a \wedge \neg b \vee \neg a \wedge b$. So, we can represent the original DAG as an AND-OR-NOT DAG using $\hat{n}'\leq 3\cdot \hat{n} =O(\hat{n})$ vertices. The number of edges is $m'\leq 6\cdot m=O(m)$. Due to Theorem \ref{th:and-or-dag}, we can construct a quantum algorithm with running time $O(\sqrt{\hat{n}m\log \hat{n}})$ and error probability $O(1/\hat{n})$.
\Endproof
\end{proof}

The previous theorem shows us the existence of a quantum algorithm for Boolean circuits with XOR, NAND, AND, OR and NOT gates. Let us present the result for Zhegalkin polynomial.

\begin{corollary}\label{cr:zp}
Suppose that Boolean function $F(x)$ can be represented as Zhegalkin polynomial for some integers $k,t_1,\dots,t_k$: $F(x)=ZP(x)=a\oplus\bigoplus_{i=1}^{k}C_i$, where $a\in\{0,1\},$ $C_i=\bigwedge_{z=1}^{t_i} x_{j_z}$. Then, there is a quantum algorithm for $F$ with running time $O\left(\sqrt{k\log k(k+t_1+\dots+t_k)}\right)$ and error probability $O(1/k)$.
\end{corollary}
\begin{proof}
Let us present $C_i$ as one $AND$ vertex with $t_i$ outgoing edges. $XOR$ operation is replaced by two $AND$, one $OR$ vertex and $6$ edges. So, $m=6\cdot (k-1) + t_1 +\dots + t_k=O(k+t_1+\dots+t_k)$, $n=3\cdot (k-1) + k=O(k)$.
Due to Theorem \ref{th:xor-dag}, we obtain claimed properties.
\Endproof
\end{proof}
\section{The Quantum Algorithm for the Single Source Longest Path Problem for a Weighted DAG and the Diameter Search Problem for Unweighted DAG}\label{sec:diameter}

In this section, we consider two problems for DAGs.
\subsection{The Single Source Longest Path Problem for Weighted DAG}
Let us apply the approach to the Single Source Longest Path problem.

Let us consider a weighted DAG $G=(V,E)$ and the weight of an edge $e=(v_i,v_j)$ is $w(i,j)$, for $i,j\in \{1,\dots,n\}, e\in E$.

Let we have a vertex $v_s$ then we should compute $t[1],\dots,t[n]$. Here $t[i]$ is the length of the longest path from $v_s$ to $v_i$. If a vertex $v_i$ is not reachable from $v_s$ then $t[i]=-\infty$.

Let us present the algorithm for longest paths lengths computing. 

Let $t=(t[1],\dots,t[n])$ be an array which stores results for vertices. Let $Q$ be the D{\"u}rr-H{\o}yer quantum algorithm for $MAX$ function. Let $\hat{Q}^{2\log_2 (n)}$ be a boosted version of $Q$  (Algorithm \ref{alg:bost-qalgo}, Lemma \ref{lm:boosting}).  
\begin{algorithm}
\caption{Quantum Algorithm for the Single Source Longest Path Search problem.}\label{alg:longest-path}
\begin{algorithmic}
\State $t\gets (-\infty,\dots,-\infty)$
\State $t[s]\gets 0$
\For{$i=s+1\dots n$}
  \State $t[i] \gets \hat{Q}^{2\log_2 n}(t[j_1]+w(i,j_1),\dots,t[j_w]+w(i,j_w))$, where $w=|D'_i|$, $(v_{j_1},\dots,v_{j_{w}})=D'_i$.
\EndFor
\State \Return $t$
\end{algorithmic}
\end{algorithm}

Algorithm \ref{alg:longest-path} has the following property:
\begin{theorem}\label{th:longest-path}
Algorithm \ref{alg:longest-path} solves the Single Source Longest Path Search problem with expected running time $O(\sqrt{nm}\log n)$ and error probability $O(1/n)$.
\end{theorem}
\begin{proof}
Let us prove the correctness of the algorithm. In fact, the algorithm computes $t[i]=\max(t[j_1]+w(i,j_1),\dots,t[j_w]+w(i,j_w))$. Assume that $t[i]$ is less than the length of the longest path. Then there is $v_z\in D_i'$ that precedes $v_i$ in the longest text. Therefore, the length of the longest path is $t[z]+w(i,z)>t[i]$. This claim contradicts with definition of $t[i]$ as maximum. The bounds for the running time and the error probability follows from Lemmas \ref{lm:dp-time}, \ref{lm:dp-err}. 
\end{proof}

\subsection{The Diameter Search Problem for a Unweighted DAG}

Let us consider an unweighted DAG $G=(V,E)$. Let $len(i,j)$ be the length of the shortest path between $v_i$ and $v_j$. If the path does not exist, then $len(i,j)=-1$.  The diameter of the graph $G$ is $diam(G)=\max\limits_{i,j\in\{1,\dots,|V|\}}len(i,j)$.  
For a given graph $G=(V,E)$, we should find the diameter of the graph.

It is easy to see that the diameter is the length of a path between a non-sink vertex and some other vertex. If this fact is false, then the diameter is $0$. 

Using this fact, we can present the algorithm. 

Let $t^z=(t^z[1],\dots,t^z[n])$ be an array which stores shortest paths from vertices to vertex $v_z\in V\backslash L$. Let $Q$ be the D{\"u}rr-H{\o}yer quantum algorithm for the $MIN$ function. Let $\hat{Q}^{2\log_2 (n)}$ be a boosted version of $Q$  (Algorithm \ref{alg:bost-qalgo}, Lemma \ref{lm:boosting}). Let $Q_{max}$ and $\hat{Q}_{max}^{\log_2 (n)}$ be the quantum algorithm and the boosted version of the algorithm for the $MAX$ function that ignore $+\infty$ values.  
\begin{algorithm}
\caption{Quantum Algorithm for the Diameter Search problem.}\label{alg:diam}
\begin{algorithmic}
\For{$z=\hat{n}\dots 1$}
\State $t^z\gets (+\infty,\dots,+\infty)$
\State $t^z[z]\gets 0$
\For{$i=z+1\dots n$}
  \State $t^z[i] \gets \hat{Q}^{2\log_2 n}(t^z[j_1],\dots,t^z[j_w])+1$, where $w=|D'_i|$, $(v_{j_1},\dots,v_{j_{w}})=D'_i$.
\EndFor
\EndFor
\State $diam(G)=\hat{Q}_{max}^{\log_2 (n)}(0,t^1[2],\dots,t^1[n],t^2[3],\dots,t^2[n],\dots,t^{\hat{n}}[\hat{n}+1],\dots,t^{\hat{n}}[n])$
\State \Return $diam(G)$
\end{algorithmic}
\end{algorithm}

Algorithm \ref{alg:diam} has the following property:
\begin{theorem}\label{th:diam}
Algorithm \ref{alg:diam} solves the Diameter Search problem with expected running time $O(\hat{n}(n+\sqrt{nm})\log n)$ and error probability $O(1/n)$.
\end{theorem}
\begin{proof}
The correctness of the algorithm can be proven similar to the proof of Theorem \ref{th:longest-path}. The bounds for the running time and the error probability follows from Lemmas \ref{lm:dp-time}, \ref{lm:dp-err}. 
\end{proof}

{\bf Acknowledgements.}
This work was supported by Russian Science Foundation Grant 17-71-10152.

A part of work was done when K. Khadiev visited University of Latvia. We thank Andris Ambainis, Alexander Rivosh and Aliya Khadieva for help and useful discussions.

\bibliographystyle{plain}
\bibliography{tcs}

\begin{thebibliography}{10}

\bibitem{aakv2018}
F.~Ablayev, M.~Ablayev, K.~Khadiev, and A.~Vasiliev.
\newblock Classical and quantum computations with restricted memory.
\newblock {\em LNCS}, 11011:129--155, 2018.

\bibitem{aakk2018}
F.~Ablayev, A.~Ambainis, K.~Khadiev, and A.~Khadieva.
\newblock Lower bounds and hierarchies for quantum memoryless communication
  protocols and quantum ordered binary decision diagrams with repeated test.
\newblock {\em In SOFSEM, LNCS}, 10706:197--211, 2018.

\bibitem{agky16}
F.~Ablayev, A.~Gainutdinova, K.~Khadiev, and A.~Yakary{\i}lmaz.
\newblock Very narrow quantum \mbox{OBDD}s and width hierarchies for classical
  \mbox{OBDD}s.
\newblock {\em Lobachevskii Journal of Mathematics}, 37(6):670--682, 2016.

\bibitem{agky14}
F.~Ablayev, A.~Gainutdinova, K.~Khadiev, and A.~Yakaryılmaz.
\newblock Very narrow quantum \mbox{OBDD}s and width hierarchies for classical
  \mbox{OBDD}s.
\newblock In {\em DCFS}, volume 8614 of {\em LNCS}, pages 53--64. Springer,
  2014.

\bibitem{av2009}
F.~Ablayev and A.~Vasilyev.
\newblock On quantum realisation of boolean functions by the fingerprinting
  technique.
\newblock {\em Discrete Mathematics and Applications}, 19(6):555--572, 2009.

\bibitem{an2009}
A.~Ambainis and N.~Nahimovs.
\newblock Improved constructions of quantum automata.
\newblock {\em Theoretical Computer Science}, 410(20):1916--1922, 2009.

\bibitem{a2007}
Andris Ambainis.
\newblock A nearly optimal discrete query quantum algorithm for evaluating nand
  formulas.
\newblock {\em arXiv preprint arXiv:0704.3628}, 2007.

\bibitem{a2010}
Andris Ambainis.
\newblock Quantum algorithms for formula evaluation.
\newblock {\em https://arxiv.org/abs/1006.3651}, 2010.

\bibitem{a2017}
Andris Ambainis.
\newblock Understanding quantum algorithms via query complexity.
\newblock {\em arXiv preprint arXiv:1712.06349}, 2017.

\bibitem{avrz2010}
Andris Ambainis, Andrew~M Childs, Ben~W Reichardt, Robert {\v{S}}palek, and
  Shengyu Zhang.
\newblock Any and-or formula of size n can be evaluated in time $n^{1/2}+o(1)$
  on a quantum computer.
\newblock {\em SIAM Journal on Computing}, 39(6):2513--2530, 2010.

\bibitem{as2006}
Andris Ambainis and Robert {\v{S}}palek.
\newblock Quantum algorithms for matching and network flows.
\newblock In {\em Annual Symposium on Theoretical Aspects of Computer Science},
  pages 172--183. Springer, 2006.

\bibitem{ab2009}
Sanjeev Arora and Boaz Barak.
\newblock {\em Computational complexity: a modern approach}.
\newblock Cambridge University Press, 2009.

\bibitem{bbht98}
Michel Boyer, Gilles Brassard, Peter H{\o}yer, and Alain Tapp.
\newblock Tight bounds on quantum searching.
\newblock {\em Fortschritte der Physik}, 46(4-5):493--505, 1998.

\bibitem{bhmt2002}
Gilles Brassard, Peter H{\o}yer, Michele Mosca, and Alain Tapp.
\newblock Quantum amplitude amplification and estimation.
\newblock {\em Contemporary Mathematics}, 305:53--74, 2002.

\bibitem{bkt2018}
Mark Bun, Robin Kothari, and Justin Thaler.
\newblock Quantum algorithms and approximating polynomials for composed
  functions with shared inputs.
\newblock {\em arXiv preprint arXiv:1809.02254}, 2018.

\bibitem{ckk2012}
Andrew~M Childs, Shelby Kimmel, and Robin Kothari.
\newblock The quantum query complexity of read-many formulas.
\newblock In {\em European Symposium on Algorithms}, pages 337--348. Springer,
  2012.

\bibitem{cormen2001}
Thomas~H Cormen, Charles~E Leiserson, Ronald~L Rivest, and Clifford Stein.
\newblock {\em Introduction to Algorithms-Secund Edition}.
\newblock McGraw-Hill, 2001.

\bibitem{dw2001}
Ronald De~Wolf.
\newblock {\em Quantum computing and communication complexity}.
\newblock 2001.

\bibitem{d2008phd}
Sebastian D{\"o}rn.
\newblock {\em Quantum complexity of graph and algebraic problems}.
\newblock PhD thesis, Universit{\"a}t Ulm, 2008.

\bibitem{d2009}
Sebastian D{\"o}rn.
\newblock Quantum algorithms for matching problems.
\newblock {\em Theory of Computing Systems}, 45(3):613--628, 2009.

\bibitem{dhhm2004}
Christoph D{\"u}rr, Mark Heiligman, Peter H{\o}yer, and Mehdi Mhalla.
\newblock Quantum query complexity of some graph problems.
\newblock In {\em International Colloquium on Automata, Languages, and
  Programming}, pages 481--493. Springer, 2004.

\bibitem{dhhm2006}
Christoph D{\"u}rr, Mark Heiligman, Peter H{\o}yer, and Mehdi Mhalla.
\newblock Quantum query complexity of some graph problems.
\newblock {\em SIAM Journal on Computing}, 35(6):1310--1328, 2006.

\bibitem{dh96}
Christoph Durr and Peter H{\o}yer.
\newblock A quantum algorithm for finding the minimum.
\newblock {\em arXiv preprint quant-ph/9607014}, 1996.

\bibitem{gg85}
Semyor~Grigor'evich Gindikin, Simon Gindikin, and Semen~G Gindikin.
\newblock {\em Algebraic logic}.
\newblock Springer Science \& Business Media, 1985.

\bibitem{g96}
Lov~K Grover.
\newblock A fast quantum mechanical algorithm for database search.
\newblock In {\em Proceedings of the twenty-eighth annual ACM symposium on
  Theory of computing}, pages 212--219. ACM, 1996.

\bibitem{gr2005}
Lov~K Grover and Jaikumar Radhakrishnan.
\newblock Is partial quantum search of a database any easier?
\newblock In {\em Proceedings of the seventeenth annual ACM symposium on
  Parallelism in algorithms and architectures}, pages 186--194. ACM, 2005.

\bibitem{ikpy2018}
R.~Ibrahimov, K.~Khadiev, K.~Pr\={u}sis, and A.~Yakaryılmaz.
\newblock Error-free affine, unitary, and probabilistic \mbox{OBDD}s.
\newblock {\em Lecture Notes in Computer Science}, 10952 LNCS:175--187, 2018.

\bibitem{quantumzoo}
Stephen Jordan.
\newblock Bounded error quantum algorithms zoo.
\newblock https://math.nist.gov/quantum/zoo.

\bibitem{kk2017}
K.~Khadiev and A.~Khadieva.
\newblock Reordering method and hierarchies for quantum and classical ordered
  binary decision diagrams.
\newblock In {\em CSR 2017}, volume 10304 of {\em LNCS}, pages 162--175.
  Springer, 2017.

\bibitem{kkm2018}
K.~Khadiev, A.~Khadieva, and I.~Mannapov.
\newblock Quantum online algorithms with respect to space and advice
  complexity.
\newblock {\em Lobachevskii Journal of Mathematics}, 39(9):1210--1220, 2018.

\bibitem{l2009}
Fran{\c{c}}ois Le~Gall.
\newblock Exponential separation of quantum and classical online space
  complexity.
\newblock {\em Theory of Computing Systems}, 45(2):188--202, 2009.

\bibitem{nc2010}
Michael~A Nielsen and Isaac~L Chuang.
\newblock {\em Quantum computation and quantum information}.
\newblock Cambridge university press, 2010.

\bibitem{y2003}
Sergey~Vsevolodovich Yablonsky.
\newblock {\em Introduction to Discrete Mathematics: Textbook for Higher
  Schools}.
\newblock Mir Publishers, 1989.

\bibitem{z1927}
Ivan Zhegalkin.
\newblock On the technique of calculating propositions in symbolic logic (sur
  le calcul des propositions dans la logique symbolique).
\newblock {\em Matematicheskii Sbornik (in Russian and French)}, 34(1):9--28,
  1927.

\end{thebibliography}





\end{document}